 \documentclass[a4paper,11pt]{scrartcl}

\usepackage{microtype}

\usepackage{comment}

\usepackage{amsfonts}

\usepackage{etoolbox}
\newtoggle{ARXIV}
 \toggletrue{ARXIV}

\iftoggle{ARXIV}{
\usepackage[utf8]{inputenc}
\usepackage{amsthm}
\theoremstyle{plain}
\newtheorem{theorem}{Theorem}
\newtheorem{definition}[theorem]{Definition}
\newtheorem{lemma}[theorem]{Lemma}

\newtheorem{corollary}[theorem]{Corollary}
}{}

\newcommand{\fpt}{\mathsf{FPT}}
\newcommand{\ptime}{\mathsf{PTIME}}
\newcommand{\fptime}{\mathsf{fPTIME}}
\newcommand{\np}{\mathsf{NP}}
\newcommand{\conp}{\mathsf{co\textup{-}NP}}
\newcommand{\nc}{\mathsf{NC}}

\newcommand{\fancyc}{\mathcal{C}}
\newcommand{\fancyd}{\mathcal{D}}
\newcommand{\fancyl}{\mathcal{L}}

\newcommand{\para}[1]{\mathsf{para\textup{-}}#1}

\newcommand{\chopped}[1]{\mathsf{chopped\textup{-}}#1}
\newcommand{\mixed}[1]{\mathsf{mixed\textup{-}}#1}

\newcommand{\mathcomp}{\mathsf{\textup{-}comp\textup{-}}}
\newcommand{\lcompc}{\fancyl\mathcomp\fancyc}
\newcommand{\polycomp}[1]{\mathsf{poly\textup{-}comp\textup{-}}#1}
\newcommand{\expcomp}[1]{\mathsf{exp\textup{-}comp\textup{-}}#1}

\newtheorem{prop}[theorem]{Proposition}
\newtheorem{remrk}[theorem]{Remark}
\newtheorem{assumption}[theorem]{Assumption}

\newcommand{\wpfin}{\wp_{\mathrm{fin}}}

\newcommand{\len}{\mathrm{len}}
\newcommand{\un}{\mathrm{un}}

\newcommand{\CI}{\mathrm{CI}}
\newcommand{\MMC}{\mathrm{MMC}}
\newcommand{\CMI}{\mathrm{CMI}}

\newcommand{\nuclass}[1]{| | \to #1}

\newcommand{\relb}{\mathbf{B}}

\newcommand{\circuitsat}{\mathrm{CIRCUIT\textup{-}SAT}}
\newcommand{\threesat}{\mathrm{3\textup{-}SAT}}
\newcommand{\hampath}{\mathrm{HAM\textup{-}PATH}}
\newcommand{\dhittingset}{$d$\mathrm{\textup{-}HITTING\textup{-}SET}}
\newcommand{\uepmc}{\mathrm{unary\textup{-}EP\textup{-}MC}}

\newcommand{\N}{\mathbb{N}}

\title{Parameter Compilation}

\iftoggle{ARXIV}
{  
\author{
Hubie Chen \footnote{Universidad del País Vasco, E-20018 San Sebastián, Spain,
\emph{and} IKERBASQUE, Basque Foundation for Science, E-48011 Bilbao, Spain
}}
}
{
\titlerunning{Parameter Compilation} 

\author[1]{Hubie Chen}
\affil[1]{
Universidad del País Vasco,
E-20018 San Sebastián,
Spain,
\emph{and} IKERBASQUE, Basque Foundation for Science,
E-48011 Bilbao,
Spain}
\authorrunning{H. Chen} 

\Copyright{Hubie Chen}


\serieslogo{}
\volumeinfo
  {Billy Editor and Bill Editors}
  {2}
  {Conference title on which this volume is based on}
  {1}
  {1}
  {1}
\EventShortName{}
\DOI{10.4230/LIPIcs.xxx.yyy.p}

}

\begin{document}

\maketitle

\begin{abstract}
In resolving instances of a computational problem, 
if multiple instances of interest share a feature in common,
it may be fruitful to compile this feature into a format
that allows for more efficient resolution,
even if the compilation is relatively expensive.
In this article, we introduce a formal framework for
classifying problems according to their compilability.
The basic object in our framework is that of
a parameterized problem, which here is a language
along with a parameterization---a map which provides,
for each instance, a so-called parameter
on which compilation may be performed.
Our framework is positioned within the 
paradigm of parameterized complexity,
and our notions are relatable to established
concepts in the theory of parameterized complexity. 
Indeed, we view our framework as 
playing a unifying role, integrating together
parameterized complexity and compilability theory.
 \end{abstract}

\section{Introduction}

In resolving instances of a computational problem, 
if it is the case that multiple instances of interest
share a feature in common, 
it may be fruitful to 
\emph{compile} this feature into a format 
that allows for more efficient resolution,
even if the compilation is relatively expensive.
As a first, simple example, 
consider the problem of deciding if two nodes of an undirected
graph are connected.  If it is anticipated
that many such connectivity queries
will share the same graph $G$, it may be worthwhile to 
compile $G$ into a format that
will allow for accelerated
resolution of the queries.
As a second example, consider the problem 
of evaluating a database query on a database.
If one is interested in 
a small set of queries that will be posed to numerous databases,
it may be worthwhile to compile the queries of interest
into a format that allows for the fastest evaluation.
Note that a relatively expensive compilation process may be
worthwhile if its results are amortized by repeated use.
Indeed, one may conceive of compilation as an off-line preprocessing,
whose expense is offset by its later on-line use.

In this article, we attempt to make an infrastructural
contribution by introducing a formal framework 
for classifying problems according to their compilability.
Such a framework was previously presented by
Cadoli, Domini, Liberatore, and Schaerf~\cite{CDLS02-preprocessing},
(hereafter, \emph{CDLS}); we will discuss the relationship
between our framework and theirs below.

The basic object in our framework is a
\emph{paramaterized problem}, which we define
to be a language $Q$ along with a \emph{parameterization $\kappa$},
a polynomial-time computable mapping defined 
from strings to strings.
(For precise details and justifications of definitions,
refer to the technical sections of the article.)
As usual, we refer to $\kappa(x)$
as the \emph{parameter} of an instance $x$.
In our framework, we wish to understand 
for which problems
the parameters can be succinctly compiled into a form such that,
post-compilation, the problem can be resolved in polynomial-time.
The base class of our framework, called $\polycomp{\ptime}$,
is (essentially) defined to contain a parameterized problem
$(Q, \kappa)$ if there exists a polynomial-length, computable
function $c$ such that if each instance $x$ is always presented
along with $c(\kappa(x))$, then each instance can be resolved
in polynomial time.  The function $c$ models the notion of
compilation of the parameters.
In order to give evidence of non-containment
in the class $\polycomp{\ptime}$
and also to facilitate problem classification, 
we introduce a hierarchy
of parameterized complexity classes $\chopped{\fancyc}$,
one for each classical complexity class $\fancyc$;
we observe (for example)
that $\chopped{\np}$ is not contained in
$\polycomp{\ptime}$, assuming that the polynomial hierarchy
does not collapse
(see 
Proposition~\ref{prop:chopped-ptime-equals-polycomp-ptime}
and
Theorem~\ref{thm:properness-chopped-classes}), 
and hence hardness of a problem
for $\chopped{\np}$ can be construed as evidence of
non-containment in $\polycomp{\ptime}$. 
We observe a number of completeness and hardness results
for $\chopped{\np}$ (Section~\ref{sect:completeness-hardness-chopped-np}).\footnote{
These results include a hardness result
on model checking existential positive sentences
(Proposition~\ref{prop:upemc});
we remark that obtaining a broader understanding of
the non-compilability results in the author's 
previous study of model checking~\cite{Chen14-existentialpositive} 
was in fact a motivation of the present article.
}
The class $\polycomp{\ptime}$ and the classes $\chopped{\fancyc}$
are all subsets of the parameterized class $\fpt$,
which is considered to be the basic notion of tractability
in the paradigm of parameterized complexity.\footnote{
	Note that the containment of $\polycomp{\ptime}$ in 
	$\fpt$ is essentially observed (in different language)
	in the last paragraph of Section 5 of~\cite{CDLS02-preprocessing}.
}
We believe that the introduced classes constitute 
a natural stratification of $\fpt$, whose study might well 
lead to deeper theory.

In the CDLS framework, the basic object is a language
consisting of pairs of strings (called a \emph{language of pairs}),
and one aims to understand when a compilation can be applied
to the first entry of each pair so as to allow for
efficient decision.  
This is a point of difference with our framework, but note that
the notions from our framework can be readily applied
to the languages of pairs that CDLS study by using the
parameterization $\pi_1$ that returns the first entry of a pair.
Another point of difference between our framework and theirs
is that their analog of our compilation function $c$
is not required to be computable;
while this makes
the negative results stronger, 
in our view there ought to be a focus on positive results,
which are rendered less meaningful without the
computability requirement.
(Actually, we are not aware of 
any natural computable problem for which the presence
or absence of this requirement makes a difference.)
Although these differences may appear slight,
by initiating our theory with our particular choice
of definitions,
we are able to position our framework within the 
language and tradition
of parameterized complexity and
relate our notions to existing ideas in parameterized complexity.
For instance, although not difficult,
we can directly relate the notion of a 
\emph{polynomial kernelization} to
the classes $\chopped{\fancyc}$
(Proposition~\ref{prop:polynomial-kernelization})
and use this relationship to observe 
the $\chopped{\np}$-completeness of the standard
parameterization of the 
hitting set problem for hypergraphs of bounded edge size
(see Theorem~\ref{thm:chopped-np-completeness}).
We also believe that the theory that results 
from our framework's definitions
witnesses that working with parameterized problems
as opposed to languages of pairs allows for
greater flexibility and smoother formulation 
(consider, for example, 
the characterization of $\chopped{\fancyc}$ using
the \emph{length parameterization} given by
Proposition~\ref{prop:chopped-characterization}).

Our framework and that of CDLS also differ later in the
respective developments.
Notably, 
our notion of reduction 
(Definition~\ref{def:poly-comp-reduction})
is readily
seen to be a restricted 
version of the usual fpt many-one reduction in
parameterized complexity,
and
we believe that our notion of reduction
is conceptually simpler to comprehend
than that of CDLS~\cite[Definition 2.8]{CDLS02-preprocessing}.
Despite these differences---and we view this as crucial---we
demonstrate how classification
results obtained in the CDLS framework
can be formulated and obtained in our framework;
this is made precise and performed in Section~\ref{sect:cdls}.

The presentation and development of our framework
may thus be viewed as playing a unifying role,
integrating together parameterized complexity and compilation.
Our choices of definitions and in formulation
allows us to directly relate the resulting concepts
to the theory of parameterized complexity.
At the same time, we believe that these concepts
capture in an essential way the core mathematical content 
and the core ideas of the CDLS framework
(as borne out by our results and discussion in Section~\ref{sect:cdls}).

\paragraph{\bf Related work.}
The CDLS framework was deployed
after its introduction
to analyze the compilability of reasoning tasks,
see for example~\cite{Liberatore02-size-of-mdp,LiberatoreSchaerf07-compilability-abduction}.

In the context of compiling propositional formulas,
a notion of compilation whereby a compiled version should
have the same models as the original formula was studied,
for example
by Gogic et al.~\cite{GKPS95-comparative-linguistics}
and by Darwiche and Marquis~\cite{DarwicheMarquis02-knowledge-compilation-map};
see also the recent work by Bova et al.~\cite{BCMS14-expander-cnfs-exponential-dnnf}.

Variants of the CDLS framework that relaxed the 
requirement that the size of compilations be polynomial
were also studied~\cite{Chen03-average-case-compilability,Chen05-parameterized-compilability}.

Finally, we mention that 
Fan, Geerts, and Neven~\cite{FanGeertsNeven13-big-data-preprocessing}
also developed a framework for classifying problems according
to compilability, with a focus on efficient
parallel processing (modelled using the complexity class $\nc$)
following a polynomial-time compilation.
We believe that it may be of interest to better understand
and develop
the relationship between our framework and theirs.
While we leave such a study to future work,
we mention that their notion of \emph{$\Pi$-tractability}
on a language $Q$ of pairs
can be 
described using our framework.\footnote{
   Precisely, a language $Q$ of pairs being 
\emph{$\Pi$-tractable} can be
verified to be equivalent to 
the parameterized problem $(Q, \pi_1)$ 
being in our class $\polycomp{\nc}$
via a poly-compilable function $g(x,y) = f(c(\pi_1(x,y)),(x,y)) = 
f(c(x),(x,y))$ where $c$ is polynomial-time computable.
}

\section{Preliminaries}

Throughout, $\pi_i$ denotes the operator that, given a tuple,
returns the $i$th entry of the tuple.

When $T$ is a set, we use $\wpfin(T)$ to denote the
set $\{ S \subseteq T ~|~ S \textup{ is finite} \}$.

We generally use $\Sigma$ to denote the alphabet over which
strings are formed, and generally assume $\{ 0, 1 \} \subseteq \Sigma$.
As is standard, we freely interchange between
elements of $\Sigma^*$ and $\Sigma^* \times \Sigma^*$.
When $k \geq 0$, we use $\Sigma^{\leq k}$
to denote the set of strings in $\Sigma^*$ of length
less than or equal to $k$.
For $n \in \N$, we use $\un(n)$ to denote
its unary encoding $1^n$ as a string.

We assume that languages under discussion
are non-trivial, that is, not equal to $\emptyset$ nor $\Sigma^*$.
We use $\ptime$ to denote the set of all languages decidable in 
polynomial time,
and $\fptime$ to denote the set of all functions
from $\Sigma^*$ to $\Sigma^*$ that are computable in polynomial time.

Here, by a \emph{parameterization}, we refer
to a map from $\Sigma^*$ to $\Sigma^*$.
Relative to a parameterization $\kappa: \Sigma^* \to \Sigma^*$,
it is typical to refer to $\kappa(x)$ as the \emph{parameter}  
of the string $x$.
While it is typical in the literature to define
a parameterization to be a map from $\Sigma^*$ to $\N$,
in this article we want to apply 
compilation functions to parameters and
discuss the \emph{length} of the results,
and we find that this is facilitated in many cases
by permitting the parameter of a string to be a string itself.
Throughout, we employ the following assumption
(which is discussed below in Remark~\ref{remrk:polycomp-ptime}).
\begin{assumption}
\label{assumption:param-polytime-comp}
Each parameterization is
polynomial-time computable, that is, in $\fptime$.
\end{assumption}
We use $\len$ to denote the parameterization
defined by $\len(x) = \un(|x|)$.
A \emph{parameterized problem} is a pair $(Q, \kappa)$
consisting of a language $Q$ and a parameterization $\kappa$.

By a \emph{classical complexity class}, we refer to 
a set of computable languages.
For a classical complexity class $\fancyc$,
we define $\para{\fancyc}$
to be the set that contains a parameterized problem
$(Q, \kappa)$ if there exists a computable function
$c: \Sigma^* \to \Sigma^*$, and a language 
$Q' \subseteq \Sigma^* \times \Sigma^*$ in $\fancyc$
such that, for each string $x \in \Sigma^*$, it holds that
$x \in Q \Leftrightarrow (c(\kappa(x)),x) \in Q'$.
We define $\fpt$ to be $\para{\ptime}$
(although this is perhaps not the usual definition of
$\fpt$, it is equivalent~\cite[Theorem 1.37]{FlumGrohe06-parameterizedcomplexity}).

As usual, when $\fancyd$ is a set of problems
(that is, a set of either languages or parameterized problems),
we say that a problem $P'$ is \emph{$\fancyd$-hard}
under a notion of reduction if each $P$ in $\fancyd$
reduces to $P'$; if in addition $P' \in \fancyd$,
we say that $P'$ is \emph{$\fancyd$-complete}.
We say that $\fancyd$ is \emph{closed} under a notion of reduction
if, when $P$ reduces to $P'$ and $P' \in \fancyd$,
it holds that $P \in \fancyd$.


\section{Framework}

\subsection{Problem classes}

In this subsection, we introduce the 
complexity classes of our framework.
We begin by introducing two basic definitions.
By a \emph{length function}, we refer to
a function from $\N$ to $\N$.

\begin{definition}
\normalfont
Let $\fancyl$ be a set of length functions.
\begin{itemize}

\item A function $c: \Sigma^* \to \Sigma^*$ is said to be
$\fancyl$-length if there exists $\ell \in \fancyl$
such that
for each $x \in \Sigma^*$,
it holds that $|c(x)| \leq \ell(|x|)$.

\item A function $g: \Sigma^* \to \Sigma^*$ is
\emph{$\fancyl$-compilable} 
with respect to a parameterization $\kappa$
if there exist $f \in \fptime$ and
a computable, $\fancyl$-length function $c: \Sigma^* \to \Sigma^*$
such that (for each $x \in \Sigma^*$) 
$g(x) = f(c(\kappa(x)),x)$.

\end{itemize}
\end{definition}

Put informally, a function $g$ is $\fancyl$-compilable if,
when one has the result
of applying $c$ to the parameter of an instance $x$,
the value $g(x)$ can be efficiently computed.
The function $c$ can be thought of as performing
a precomputation or compilation of the parameter.
Here, we do not place any restriction on the computational resources
needed to compute $c$, other than requiring that 
$c$ is computable.  
We view the requirement that $c$ be computable as natural
in terms of claiming positive results,
as we find it hard to argue that a non-computable compilation
would actually be usable.
We do restrict the length of $c$
according to $\fancyl$; we will be most interested in the case
where the length of $c$ is polynomially bounded.

With this terminology in hand, we can now 
define our first classes of parameterized problems.

\begin{definition}
\normalfont
Let $\fancyl$ be a set of length functions,
and let $\fancyc$ be a classical complexity class.
\begin{itemize}

\item We say that 
a parameterized problem $(Q, \kappa)$ 
is \emph{$\fancyl$-compilable to $\fancyc$}
if
there exists a  function $g: \Sigma^* \to \Sigma^*$
that is $\fancyl$-compilable (with respect to $\kappa$)
and a language $Q' \in \fancyc$
such that (for each $x \in \Sigma^*$)
$x \in Q \Leftrightarrow g(x) \in Q'$.

\item We define $\lcompc$ to be the 
set that contains each parameterized problem that
is $\fancyl$-compilable to $\fancyc$.
\end{itemize}
When $\fancyl$ is the set of all polynomials on $\N$,
we define $\polycomp{\fancyc}$ as $\lcompc$
and speak, for instance, of \emph{poly-compilability};
similarly, when $\fancyl$ is the set 
$\cup \{ O(2^{p}) ~|~ p \textup{ is a polynomial} \}$
of \emph{exponential functions},
we define $\expcomp{\fancyc}$ as $\lcompc$
and speak, for instance, of \emph{exp-compilability}.
\end{definition}

\begin{remrk}
\label{remrk:polycomp-ptime}
\normalfont
In this paper,
the smallest class that we will consider is
$\polycomp{\ptime}$, and we will regard an inclusion result
in this class as the most positive result 
demonstrable on a parameterized problem.
Suppose that a parameterized problem 
$(Q, \kappa)$ is in $\polycomp{\ptime}$
via $g(x) = f(c(\kappa(x)),x)$ and $Q'$.
One way to intuitively interpret this inclusion is as follows.
Suppose that the value $c(k)$ is known 
for parameter values $k$ in a limited range.
Then, for each instance $x \in \Sigma^*$ 
having parameter value $\kappa(x)$ in that limited range,
whether or not $x \in Q$ can be determined efficiently,
by applying the efficiently computable function
$f$ to $(c(\kappa(x)),x)$ and then by invoking an efficient
decision procedure for $Q'$.
Indeed, our intention here is to model the notion of
efficient decidability modulo knowledge of $c$;
this is why we put into effect
Assumption~\ref{assumption:param-polytime-comp}.
\end{remrk}

We observe the following upper bound on each class
$\lcompc$, which in particular indicates
that $\polycomp{\ptime} \subseteq \fpt$.

\begin{prop}
\label{prop:lcompc-in-parac}
Let $\fancyl$ be a set of length functions,
and let $\fancyc$ be a classical complexity class
that is closed under many-one polynomial-time reduction.
It holds that $\lcompc \subseteq \para{\fancyc}$.
\end{prop}

\begin{proof}
Suppose that $(Q, \kappa)$ is $\fancyl$-compilable to $\fancyc$
via $g(x) = f(c(\kappa(x)),x)$ and $Q'' \in \fancyc$,
so that $x \in Q \Leftrightarrow g(x) \in Q''$.
Define $Q' = \{ (a,b) ~|~ f(a,b) \in Q'' \}$.
The language $Q'$ is many-one polynomial-time reducible
to $Q''$ via $f$, so $Q' \in \fancyc$.
We have $x \in Q \Leftrightarrow (c(\kappa(x)),x) \in Q'$,
implying that $Q \in \para{\fancyc}$.
\end{proof}

We now define a family of complexity classes
which will be used to classify parameterized problems
in $\fpt$ according to their compilability,
and in particular to give
evidence of non-inclusion in $\polycomp{\ptime}$,
via hardness results.

\begin{definition}
For each classical complexity class $\fancyc$,
we define $\chopped{\fancyc}$ 
to be the set that contains each parameterized problem
$(Q, \kappa)$ that is in $\polycomp{\fancyc}$ via
a function $g$ for which there exists a polynomial $p: \N \to \N$
such that
(for each $x \in \Sigma^*$)
$|g(x)| \leq p(|\kappa(x)|)$.
\end{definition}

For the sake of understanding this definition,
let us call the restriction of a language $Q'$ to
$Q' \cap \Sigma^{\leq k}$ the \emph{chop having magnitude $k$}
of $Q'$.  
Then, intuitively speaking, a problem is in $\chopped{\fancyc}$
if it is in $\polycomp{\fancyc}$ via $g$ and $Q'$
where $g(x)$ accesses only a chop (of $Q'$) having magnitude
restricted by a polynomial in the parameter of $x$.
The following proposition is clear from the definition
of $\chopped{\fancyc}$.
\begin{prop}
\label{prop:choppedc-in-polycompc}
For each classical complexity class $\fancyc$, it holds that
$$\chopped{\fancyc} \subseteq \polycomp{\fancyc}.$$
\end{prop}

We also have the following upper bound on 
$\chopped{\fancyc}$, which shows that the classes
$\chopped{\fancyc}$ constitute a stratification of
the class $\fpt$.

\begin{prop}
For each classical complexity class $\fancyc$, it holds that
$$\chopped{\fancyc} \subseteq \expcomp{\ptime},$$
and hence that $\chopped{\fancyc} \subseteq \fpt$
(by Proposition~\ref{prop:lcompc-in-parac}).
\end{prop}

\begin{proof}
We prove that
$\chopped{\fancyc} \subseteq \expcomp{\ptime}$.
Fix 
$x_N, x_Y \in \Sigma^*$
and
$P \in \ptime$ such that $x_Y \in P$ and $x_N \notin P$.
Assume that $(Q, \kappa)$ is in $\chopped{\fancyc}$
via $g(x) = f(c(\kappa(x)), x)$,
the polynomial $p$,
 and $Q' \in \fancyc$.
Let $c_1^+: \Sigma^* \to \Sigma^*$ be the function 
computed by the algorithm that, given $k \in \Sigma^*$,
loops over each string $y$ in $\Sigma^{\leq p(|k|)}$
and, for each such string $y$, outputs $1$ or $0$
depending on whether or not $y \in Q'$;
thus, $|c_1^+(k)| = |\Sigma^{\leq p(|k|)}|$.
Define $c^+(k) = (c_1^+(k),c(k),k)$.
Let $f^+$ be a function
computed by a polynomial-time algorithm
that, given a string
$((d_1,d,k),x)$ where $d_1$ is a string over $\{ 0, 1 \}$
of length $|\Sigma^{\leq p(|k|)}|$,
computes $x' = f(d,x)$,
computes the bit $b$ of $d_1$ corresponding to $x'$
(whenever $|x'| \leq p(|k|)$),
and outputs $x_Y$ or $x_N$
depending on whether or not $b=1$ or $b=0$.
The function $g^+(x) = f^+(c^+(\kappa(x)),x)$
witnesses that $(Q, \kappa)$ is exp-compilable to $\ptime$:
We have that $x \in Q$
iff $f(c(\kappa(x)),x) \in Q'$
iff $f^+((c_1^+(\kappa(x),c(\kappa(x)),\kappa(x)),x) = x_Y$
iff $f^+(c^+(\kappa(x)),x) \in P$.
\end{proof}

We observe that our base class $\polycomp{\ptime}$
coincides with the
class 
$\chopped{\ptime}$, which
is the smallest class that we will consider from
the hierarchy of classes $\chopped{\fancyc}$.

\begin{prop}
\label{prop:chopped-ptime-equals-polycomp-ptime}
$\chopped{\ptime} = \polycomp{\ptime}$.
\end{prop}

\begin{proof}
The $\subseteq$ direction follows from
Proposition~\ref{prop:choppedc-in-polycompc}.
For the $\supseteq$ direction,
suppose that $(Q, \kappa) \in \polycomp{\ptime}$.
Then, there exists a function $g$
that is a poly-compilable with respect to $\kappa$
and $Q' \in \ptime$
such that $x \in Q \Leftrightarrow g(x) \in Q'$.
Fix 
$x_N, x_Y \in \Sigma^*$
and
$P \in \ptime$ such that $x_Y \in P$ and $x_N \notin P$.
Let $h \in \fptime$ be a function such that,
for all $x \in \Sigma^*$,
it holds that $x \in Q'$ implies $h(x) = x_Y$
and that $x \notin Q'$ implies $h(x) = x_N$.
Then, the mapping $h(g)$ witnesses that $(Q, \kappa)$
is in $\chopped{\ptime}$:
$$x \in Q \Leftrightarrow g(x) \in Q' \Leftrightarrow h(g(x)) \in P$$
and for all strings $x \in \Sigma^*$,
it holds that $|h(g(x))| \leq \max(|x_Y|, |x_N|)$.
\end{proof}

The classes $\chopped{\fancyc}$ can be directly related
to kernelization in the following way.
Here, we say that a parameterized problem $(Q, \kappa)$
has a \emph{polynomial kernelization}
if there exists a polynomial-time computable
function $K: \Sigma^* \to \Sigma^*$ 
and a polynomial $p: \N \to \N$
such that 
(for each $x \in \Sigma^*$)
$x \in Q \Leftrightarrow K(x) \in Q$
and $|K(x)| \leq p(|\kappa(x)|)$.

\begin{prop}
\label{prop:polynomial-kernelization}
Suppose that a parameterized problem $(Q, \kappa)$
has a polynomial kernelization and $\fancyc$
is a classical complexity class such that $Q \in \fancyc$.
Then, the problem $(Q, \kappa)$ is in $\chopped{\fancyc}$.
\end{prop}

\begin{proof}
We have that $(Q, \kappa)$ in
$\polycomp{\fancyc}$ via $K$ (the function
from the definition of polynomial kernelization),
since 
$x \in Q \Leftrightarrow K(x) \in Q$.
Moreover, it holds that there exists a polynomial $p$
such that $|K(x)| \leq p(|\kappa(x)|)$
by the definition of polynomial kernelization.
\end{proof}

\subsection{Reduction}
\label{subsect:reduction}

We now introduce a notion of reduction
for comparing the compilability of parameterized problems.

\begin{definition}
\label{def:poly-comp-reduction}
We say that a parameterized problem $(Q, \kappa)$
\emph{poly-comp reduces}
to another parameterized problem $(Q', \kappa')$
if 
there exists a function $g: \Sigma^* \to \Sigma^*$
that is poly-compilable with respect to $\kappa$
and a 
poly-length, computable function $s: \Sigma^* \to \wpfin(\Sigma^*)$
such that (for each $x \in \Sigma^*$)
it holds that $x \in Q \Leftrightarrow g(x) \in Q'$
and that $\kappa'(g(x)) \in s(\kappa(x))$.
\end{definition}

The notion of poly-comp reduction can be viewed
as a restricted version of fpt many-one reduction.
(Consider, for example,
the definition given by 
Flum and Grohe~\cite[Definition 2.1]{FlumGrohe06-parameterizedcomplexity};
the function $g$ in Definition~\ref{def:poly-comp-reduction}
can be seen to be computable by a fpt-algorithm,
and the condition on the function $s$ ensures 
that their condition (3),
when reformulated for parameterizations of the type considered here,
holds.)

Note that, in Definition~\ref{def:poly-comp-reduction},
we assume that the set $s(x)$ is given according
to a standard representation that lists the strings therein;
hence, as a consequence of the assumption that $s$ is poly-length,
the size $|s(x)|$ of $s(x)$ is bounded above by a polynomial
in $|x|$.

We have the following two basic properties of poly-comp reduction.

\begin{theorem}
\label{thm:polycomp-compatible}
For each classical complexity class $\fancyc$,
it holds that $\polycomp{\fancyc}$ is closed under
poly-comp reduction.
\end{theorem}

\begin{theorem}
\label{thm:polycomp-reducibility-transitive}
Poly-comp reducibility is transitive.
\end{theorem}

We establish these theorems in the appendix
(Section~\ref{sect:reduction-proofs}).



We now give an alternative characterization of $\chopped{\fancyc}$
in terms of poly-comp reduction.

\begin{prop}
\label{prop:chopped-characterization}
Let $\fancyc$ be a classical complexity class.
A parameterized problem $(Q, \kappa)$ is in $\chopped{\fancyc}$
if and only if
there exists a language $Q' \in \fancyc$
such that $(Q, \kappa)$ poly-comp reduces to $(Q', \len)$.
\end{prop}

\begin{proof}
For the forward direction, suppose $(Q, \kappa) \in \chopped{\fancyc}$.
There exists a function $g$ that is poly-compilable
with respect to $\kappa$,
a language $Q' \in \fancyc$,
and a polynomial $p$ such that:
for each $x \in \Sigma^*$,
it holds that $x \in Q$ iff $g(x) \in Q'$
and it holds that $|g(x)| \leq p(|\kappa(x)|)$.
Define $s(x) = \{ \un(0), \un(1), \ldots, \un(p(|\kappa(x)|)) \}$;
then, for each $x \in \Sigma^*$,
it holds that $\kappa'(g(x)) = \un(|g(x)|) \in s(x)$.
So, $(g,s)$ is a poly-comp reduction
from $(Q, \kappa)$ to $(Q', \len)$.

For the backward direction, 
suppose that $(Q, \kappa)$ poly-comp reduces
to $(Q', \len)$ via $(g,s)$, where $Q' \in \fancyc$.
We have, for each $x \in \Sigma^*$
that $x \in Q \Leftrightarrow g(x) \in Q'$
and that $\len(g(x)) \in s(\kappa(x))$.
But there exists a polynomial $p$ such that
for each $x \in \Sigma^*$ and each $y \in s(\kappa(x))$,
it holds that $|y| \leq p(|x|)$,
since $\kappa$ and $s$ are both poly-length.
Thus $\len(g(x)) = \un(|g(x)|) \in s(\kappa(x))$
 implies that
$|g(x)| = |\un(|g(x)|)| \leq p(|x|) = p(|\len(x)|)$.
\end{proof}

From the just-given characterization of
$\chopped{\fancyc}$, 
we may infer the following two results.

\begin{prop}
For each classical complexity class $\fancyc$,
the class $\chopped{\fancyc}$ is closed under poly-comp reduction.
\end{prop}

\begin{proof}
This is a consequence of 
Proposition~\ref{prop:chopped-characterization}
and
Theorem~\ref{thm:polycomp-reducibility-transitive}.
\end{proof}

\emph{When discussing a class $\chopped{\fancyc}$,
we assume by default that \emph{hardness}
 and \emph{completeness} are with respect to 
 poly-comp reducibility.}

\begin{prop}
\label{prop:complete-for-choppedc}
Let $\fancyc$ be a classical complexity class
and assume that $Q^+$ is $\fancyc$-complete
under many-one polynomial-time reduction.
Then, the parameterized problem $(Q^+, \len)$
is complete for $\chopped{\fancyc}$.
\end{prop}

\begin{proof}
The problem $(Q^+, \len)$ is in $\chopped{\fancyc}$
by Proposition~\ref{prop:chopped-characterization}.
Let $(Q, \kappa)$ be an arbitrary problem in $\chopped{\fancyc}$.
By Proposition~\ref{prop:chopped-characterization},
there exists a language $Q' \in \fancyc$
such that
$(Q, \kappa)$ poly-comp reduces to $(Q', \len)$.
By hypothesis, there exists a many-one polynomial-time
reduction from $Q'$ to $Q^+$;
it is straightforward to verify that this implies
that $(Q', \len)$ poly-comp reduces to $(Q^+, \len)$.
The result then follows from 
Theorem~\ref{thm:polycomp-reducibility-transitive}.
\end{proof}

\section{Chopped classes and advice}
\label{sect:chopped-classes}
\newcommand{\poly}{\mathrm{poly}}

In this section, we relate the classes $\chopped{\fancyc}$
to advice-based complexity classes; this will allow us
to provide evidence of separation between classes of the
form $\chopped{\fancyc}$.

We first present a known notion from computational complexity theory,
the notion of an advice version of a complexity class.
For each classical complexity class $\fancyc$,
we define $\fancyc/\poly$
to be the set that contains a language $Q$
if and only if there exists 
a poly-length map $a: \{ 1 \}^* \to \Sigma^*$
and a language $Q' \in \fancyc$
such that, for each $x \in \Sigma^*$,
it holds that $x \in Q \Leftrightarrow (a(\un(|x|),x)) \in Q'$.

The following theorem shows that containment
of one chopped class in another implies a
containment in classical complexity.

\begin{theorem}
\label{thm:comparing-choppedc}
Let $\fancyc$ and $\fancyc'$ be classical complexity classes
where $\fancyc'$ is closed under many-one polynomial-time reduction.
If $\chopped{\fancyc} \subseteq \chopped{\fancyc'}$,
then
$\fancyc \subseteq \fancyc'/\poly$.
\end{theorem}

To prove this theorem, we first establish a lemma.

\begin{lemma}
\label{lemma:chopped-to-advice}
Let $\fancyc'$ be a classical complexity class
that is closed under many-one polynomial-time reduction.
If $Q$ is a language such that
$(Q, \len) \in \chopped{\fancyc'}$,
then $Q \in \fancyc'/\poly$.
\end{lemma}

\begin{proof}
By hypothesis, there exists a language $Q' \in \fancyc'$
and a function $g: \Sigma^* \to \Sigma^*$
that is poly-compilable with respect to $\len$
such that $x \in Q \Leftrightarrow g(x) \in Q'$.
Let us denote $g(x) = f(c(\len(x)),x)$.
Let $P$ be the set of pairs
$\{ (d,x) \in \Sigma^* \times \Sigma^* ~|~ f(d,x) \in Q' \}$.
Then $P$ many-one polynomial-time reduces to $Q'$ via $f$,
so by hypothesis $P \in \fancyc$.
Define $a: \{1\}^* \to \Sigma^*$
by $a(1^n) = c(1^n)$.
We have $Q \in \fancyc'/\poly$ via $a$ and $Q'$,
as
$x \in Q$
iff
$f(a(\un(|x|)),x) = f(c(\len(x)),x) \in Q'$
iff
$(a(\un(|x|)),x) \in P$.
\end{proof}

\begin{proof}
(Theorem~\ref{thm:comparing-choppedc})
Suppose that $Q \in \fancyc$.
By Proposition~\ref{prop:chopped-characterization},
it holds that $(Q, \len)$ is in $\chopped{\fancyc}$.
By hypothesis,
it holds that $(Q, \len)$ is in $\chopped{\fancyc'}$.
By Lemma~\ref{lemma:chopped-to-advice},
it follows that $Q \in \fancyc'/\poly$.
\end{proof}

\newcommand{\Sigmap}{\Sigma^{\mathsf{p}}}
\newcommand{\Pip}{\Pi^{\mathsf{p}}}

We use $\Sigmap_i$ and $\Pip_i$ (with $i \geq 0$)
to denote the classes of the polynomial hierarchy (PH);
recall that $\Sigmap_0 = \Pip_0 = \ptime$,
$\Sigmap_1 = \np$, and $\Pip_1 = \conp$.
For each $i \geq 0$, let us say that the classes $\Sigma_i$
and $\Pi_i$ are at the \emph{$i$th level} of the PH.
Let us say that a class $\fancyc'$ of the PH
is \emph{above} another class $\fancyc$ of the PH
if 
they are equal 
or if
the level $j$ of $\fancyc'$ is strictly greater
than the level $i$ of $\fancyc$.

\begin{theorem}
\label{thm:properness-chopped-classes}
(follows from~\cite{Yap83-non-uniform})
Suppose that $\fancyc$ and $\fancyc'$ are classes of the PH
such that $\fancyc'$ is not above $\fancyc$.
\begin{itemize}

\item If $\chopped{\fancyc} \subseteq \chopped{\fancyc'}$,
then the PH collapses.

\item A parameterized problem $(Q, \kappa)$
that is $\chopped{\fancyc}$-hard
is not in $\chopped{\fancyc'}$,
unless the PH collapses.

\end{itemize}
\end{theorem}

\begin{proof}
For the first claim,
it follows from
Theorem~\ref{thm:comparing-choppedc}
that
$\fancyc \subseteq \fancyc'/\poly$;
by~\cite{Yap83-non-uniform}, it follows that the PH collapses.
For the second claim, observe that 
if $(Q, \kappa)$ is $\chopped{\fancyc}$-hard,
then $(Q, \kappa) \in \chopped{\fancyc'}$
implies that $\chopped{\fancyc} \subseteq \chopped{\fancyc'}$,
by the closure of $\chopped{\fancyc'}$ under poly-comp
reduction (Theorem~\ref{thm:polycomp-compatible}).
\end{proof}

\section{Relationship to the CDLS framework}
\label{sect:cdls}

In this section, we discuss the relationship between our
framework and the CDLS framework.  We in particular
show that, in a sense that we make precise, 
the completeness results that they obtain for their
problem classes can be formulated and obtained in our framework.
Note that the proofs of 
Theorems~\ref{thm:mixed-equals-chopped},
\ref{thm:hardness-from-epsilon},
and~\ref{thm:hardness-from-poly-bd-slices}
are deferred to the appendix 
(Section~\ref{sect:proofs-cdls}).

By a \emph{language of pairs},
we refer to a subset of $\Sigma^* \times \Sigma^*$.

The CDLS framework defines, 
for each classical complexity class, 
a class which they refer to as
the \emph{class of problems non-uniformly compilable
to a class $C$}, and which contains 
languages of pairs~\cite[Definition 2.7]{CDLS02-preprocessing}.
We give the following formulation of this definition.

\begin{definition}
\label{def:mixedc}
A language $B \subseteq \Sigma^* \times \Sigma^*$
of pairs is in $\mixed{\fancyc}$
if there exists a poly-length, computable
function $f: \Sigma^* \times \Sigma^* \to \Sigma^*$
and a language $B' \in \fancyc$ of pairs
such that 
$(x,y) \in B \Leftrightarrow
(f(x,\un(|y|)),y) \in B'$.
\end{definition}

Note that our definition is not exactly equivalent to theirs;
we require that the function $f$ is computable, while 
they do not.  We do not know of any natural language of pairs
for which this makes a difference; assuming computability of $f$
will allow us to more readily relate the defined classes to
those of our framework.

To illustrate how classification results on languages
obtained in the CDLS framework can be obtained in our framework,
we discuss three running examples (studied in~\cite{CDLS02-preprocessing}):

\begin{itemize}

\item Define $\CI$ \emph{(clause inference)}
to be the set of pairs $(x,y)$ 
where $x$ is a propositional 3CNF formula,
$y$ is a clause, and $x \models y$.
We assume here that clauses do not contain repeated literals.

\item Define $\MMC$ \emph{(minimal model checking)}
to be the set of pairs $(x,y)$
where $x$ is a propositional formula
and $y$ is a minimal model of $x$.
By \emph{minimal},
we mean with respect to the order $\leq$
where $z \leq z'$ if and only if all variables
true under $z$ are also true under $z'$.

\item Define $\CMI$ \emph{(clause minimal inference)}
to be the set of pairs $(x,y)$
where $x$ is a propositional formula
and $y$ is a clause that is satisfied
by all minimal models of $x$.

\end{itemize}
It is known and straightforward to verify that
$\CI, \MMC \in \conp$ and $\CMI \in \Pip_2$.
It follows immediately that
$\CI, \MMC \in \mixed{\conp}$
and $\CMI \in \mixed{\Pip_2}$.

Let us say that a parameterized problem
$(Q, \kappa)$ has \emph{poly-bounded slices} if
there exists a polynomial $p$ such that,
for each $x \in Q$,
it holds that $|x| \leq p(|\kappa(x)|)$.
Each of the three parameterized problems 
$(\CI,\pi_1)$, $(\MMC,\pi_1)$, and $(\CMI,\pi_1)$
have poly-bounded slices (as is readily verified),
and it can consequently be verified that
$(\CI,\pi_1), (\MMC,\pi_1) \in \chopped{\conp}$
and that $(\CMI,\pi_1) \in \chopped{\Pip_2}$.
It is indeed a general fact that 
when $B$ is a language of pairs 
where $(B, \pi_1)$ has poly-bounded slices,
the classes $\mixed{\fancyc}$
and $\chopped{\fancyc}$ coincide,
as made precise by the following theorem.

\begin{theorem}
\label{thm:mixed-equals-chopped}
Let $\fancyc$ be a classical complexity class
closed under many-one polynomial-time reduction.
Let $B$ be a language of pairs
such that
$(B, \pi_1)$ has poly-bounded slices.
Then, $B$ is in $\mixed{\fancyc}$
if and only if $(B, \pi_1)$ is in $\chopped{\fancyc}$.
\end{theorem}

We now present a formulation of the notion 
of reduction used in the CDLS framework
(see~\cite[Definition 2.8]{CDLS02-preprocessing}).

\begin{definition}
Let $A$ and $B$ be languages of pairs.
A \emph{mixed reduction} from $A$ to $B$
is a triple $(f_1, f_2, g)$ 
of mappings from $\Sigma^* \times \Sigma^*$ to $\Sigma$
where $f_1$ and $f_2$ are poly-length and computable,
and $g$ is polynomial-time computable,
such that
$(x,y) \in A \Leftrightarrow
(f_1(x,\un(|y|)),g(f_2(x,\un(|y|)),y)) \in B$.
\end{definition}

In analogy to Definition~\ref{def:mixedc},
here we require that the functions $f_1$ and $f_2$ are
computable.

As a way of showing hardness, 
CDLS present mixed-reductions from
languages of the form $\{ \epsilon \} \times Q^+$
where $Q^+$ is a classical language that is hard.
For example, they present the following reductions.

\begin{theorem}
\label{thm:mixed-red-to-ci}
(follows from~\cite[Proof of Theorem 2.10]{CDLS02-preprocessing})
There exists a $\conp$-complete problem $Q^+$
such that there
exists a mixed-reduction from $\{\epsilon\} \times Q^+$
to $\CI$.
\end{theorem}

\begin{theorem}
\label{thm:mixed-red-to-cmi}
(follows from~\cite[Proof of Theorem 3.2]{CDLS02-preprocessing})
There exists a $\Pip_2$-complete problem $Q^+$
such that there
exists a mixed-reduction from $\{\epsilon\} \times Q^+$
to $\CMI$.
\end{theorem}

We now present a general theorem showing that
exhibiting a reduction from 
a language of the form $\{ \epsilon \} \times Q^+$
yields a hardness result with respect to the classes
$\chopped{\fancyc}$, made precise as follows.

\begin{theorem}
\label{thm:hardness-from-epsilon}
Suppose that $A$ and $B$ are languages of pairs
such that there exists a mixed reduction from $A$ to $B$,
and let $\fancyc$ be a classical complexity class.
If $A = \{ \epsilon \} \times Q^+$
where $Q^+$ is $\fancyc$-complete,
then $(B, \pi_1)$ is $\chopped{\fancyc}$-hard.
\end{theorem}

\begin{corollary}
\label{cor:chopped-hardness-ci-cmi}
The problem $(\CI, \pi_1)$ is $\chopped{\conp}$-hard;
the problem $(\CMI, \pi_1)$ is $\chopped{\Pip_2}$-hard.
\end{corollary}

\begin{proof}
Follows from Theorems~\ref{thm:mixed-red-to-ci},
\ref{thm:mixed-red-to-cmi}
and~\ref{thm:hardness-from-epsilon}.
\end{proof}

The other way in which CDLS show hardness is by 
presenting a mixed-reduction from a problem that 
has poly-bounded slices.  For example, they prove the following.

\begin{theorem}
\label{thm:mixed-reduction-from-ci-to-mmc}
(follows from~\cite[Proof of Theorem 3.1]{CDLS02-preprocessing})
There exists a mixed-reduction from $\CI$ to $\MMC$.
\end{theorem}

We show that this form of reduction can be interpreted
as a poly-comp reduction, made precise as follows.

\begin{theorem}
\label{thm:hardness-from-poly-bd-slices}
Suppose that $A$ and $B$ are languages of pairs
such that there exists a mixed reduction from $A$ to $B$.
If $(A, \pi_1)$ has poly-bounded slices,
then $(A, \pi_1)$ poly-comp reduces to $(B, \pi_1)$.
\end{theorem}

\begin{corollary}
\label{cor:reduction-and-hardness-of-mmc}
There exists a poly-comp reduction from
$(\CI,\pi_1)$ to $(\MMC,\pi_1)$,
and hence
(by Corollary~\ref{cor:chopped-hardness-ci-cmi})
the problem $(\MMC,\pi_1)$ is $\chopped{\conp}$-hard.
\end{corollary}

\begin{proof}
Immediate from Theorems~\ref{thm:mixed-reduction-from-ci-to-mmc}
and~\ref{thm:hardness-from-poly-bd-slices}.
\end{proof}

At this point, we can observe that the non-compilability
results that CDLS obtain 
can be obtained in our framework.
For example, consider the following.
As we have seen
(and as stated in
Corollaries~\ref{cor:chopped-hardness-ci-cmi}
and~\ref{cor:reduction-and-hardness-of-mmc}),
the problems $(\CI,\pi_1)$ and $(\MMC,\pi_1)$
are $\chopped{\conp}$-hard.
This implies that these two problems
are not in $\chopped{\ptime}$, unless the PH collapses,
via Theorem~\ref{thm:properness-chopped-classes}.
We can also obtain the 
non-compilability results in (essentially) the
form stated by CDLS:
by invoking Theorem~\ref{thm:mixed-equals-chopped},
it immediately follows that the problems $\CI$ and $\MMC$
are not in $\mixed{\ptime}$, unless the PH collapses.
We want to emphasize here that the 
hardness proofs can be carried out 
using the notions and concepts of our framework.

\section{Completeness and hardness for $\chopped{\np}$}
\label{sect:completeness-hardness-chopped-np}

In this section, we present completeness and hardness results
for the class $\chopped{\np}$.

Define $\hampath$ to be the problem of deciding, given
an undirected graph $G$, whether or not $G$
contains a Hamiltonian path; define the parameterization
$\gamma$ so that $\gamma(G)$ is equal to the number of nodes in $G$.
The problems $\threesat$ and $\circuitsat$ are defined as usual.
In the context of $\threesat$, $\nu$ is the parameterization
that returns, given a formula $\phi$, the number of variables
that appear in $\phi$.
In the context of $\circuitsat$, $\mu + \nu$
is the parameterization that returns, given a circuit $C$,
the sum of 
the number of non-input gates
and 
the number of input gates 
of $C$.
For each $d \geq 2$,
we consider $\dhittingset$ to be the problem
where an instance is a pair $(H, k)$ consisting of a number $k \geq 1$
and
a hypergraph $H$ where each edge has size less than or equal to $d$,
and one is to decide whether or not $H$ has a hitting set
of size less than or equal to $k$.
Note that here, all numbers are represented in unary.

\begin{theorem}
\label{thm:chopped-np-completeness}
The following problems are $\chopped{\np}$-complete:
\begin{enumerate}

\item $(\hampath, \gamma)$
\item $(\threesat, \nu)$
\item $(\circuitsat, \mu + \nu)$
\item $(\dhittingset, \pi_2)$, for each $d \geq 2$

\end{enumerate}
\end{theorem}

\begin{proof}
We first discuss the $\chopped{\np}$-completeness of
$(\hampath, \gamma)$.
We prove that $(\hampath, \gamma)$ and $(\hampath, \len)$
are poly-comp interreducible, which suffices by the
$\np$-completeness of $\hampath$ and 
Proposition~\ref{prop:complete-for-choppedc}.
We argue that $(\hampath, \gamma)$ poly-comp reduces
to $(\hampath, \len)$ 
as follows.   
There exists a polynomial $p$
such that the size of an instance $G$ with $n$ vertices
can be bounded by $p(n)$.
We may thus take $g$ to be the identity map
and define $s(\un(n)) = \{ \un(0), \un(1), \ldots, \un(p(n)) \}$.
We argue that $(\hampath, \len)$ poly-comp reduces
to $(\hampath, \gamma)$ as follows.
Under a standard representation of graphs, 
we have that (for each graph $G$) $\gamma(G) \leq |G|$.
We may thus take $g$ to be the identity map
and define $s(\ell) = \{ \un(0), \un(1), \ldots, \un(\ell) \}$.

A similar argument can be used to show that
$(\circuitsat, \mu + \nu)$ and $(\circuitsat, \len)$
are poly-comp interreducible.
A similar argument can be used to show that
$(\threesat, \nu)$ and $(\threesat, \len)$
are poly-comp interreducible, but note the following.
For the reduction from $(\threesat, \nu)$
to $(\threesat, \len)$, to
obtain a bound of the form $p(n)$ (where $p$ is a polynomial)
on the size of an instance with $n$ variables,
instead of defining $g$ to be the identity map,
we may define $g$ to be the polynomial-time computable
function that reduces duplicate clauses,
so that there is indeed a bound of the claimed form,
since the number of clauses will then be $O(n^3)$.

The $\chopped{\np}$-hardness of the problems
$(\dhittingset, \pi_2)$
can be observed as follows.  
It suffices to prove hardness in the case of $d=2$,
as for higher values of $d$ the problem is more general.
To prove hardness in the case of $d=2$,
we take $g$ to be the
reduction given by the proof of 
\cite[Theorem 7.44]{Sipser12-introduction}
which is a many-one polynomial-time reduction
from $\threesat$ to the vertex cover problem;
since the value of $k$ in the created instance $(G,k)$
is a polynomial function of the number of variables 
of the original $\threesat$ instance, this $g$
can be completed to be a poly-comp reduction
from $(\threesat, \nu)$.
The containment of $(\dhittingset, \pi_2)$
in $\chopped{\np}$ can be obtained from 
the known polynomial kernelization 
of this problem~\cite[Section 9.1]{FlumGrohe06-parameterizedcomplexity}
and
Proposition~\ref{prop:polynomial-kernelization}.
\end{proof}

As a way of witnessing the utility of the presented framework,
let us discuss how one of the non-compilability results
from a previous paper~\cite{Chen14-existentialpositive} on the parameterized complexity
of model checking can be formulated within this framework.
Here, by a \emph{unary signature}, we mean a signature
containing only unary relation symbols.
Define $\uepmc$ to be the problem of deciding,
given a pair $(\phi, \relb)$ consisting of an existential
positive sentence and a finite relational structure,
each over the same unary signature,
whether or not $\phi$ evaluates to true on $\relb$
(see the paper~\cite{Chen14-existentialpositive} for
definitions and background).

\begin{prop}
\label{prop:upemc}
The parameterized problem $(\uepmc, \pi_1)$
is $\chopped{\np}$-hard.
\end{prop}

\begin{proof}
Let $h$ be the reduction given in \cite{Chen14-existentialpositive},
which is a many-one polynomial-time reduction
from the CNF satisfiability problem
to $\uepmc$ where an instance having $n$ variables
and $m$ clauses is mapped to an instance of the form
$(S^m_n, \relb)$, where each $S^m_n$ is a sentence.
Let $g$ be the map that,
given a 3-SAT formula $\phi$,
 eliminates duplicate clauses from $\phi$
 and then maps the result under $h$.
 For a 3-SAT formula $\phi$ with $n$ variables,
 it will thus hold that 
there exists a polynomial $C \in O(n^3)$
such that
 $\pi_1(g(\phi)) \in \{ S_n^0, S_n^1, \ldots, S_n^{C(n)} \}$.
If we define $s(n) = \{ S_n^0, S_n^1, \ldots, S_n^{C(n)} \}$,
we thus have that $(g,s)$ is a poly-comp reduction
from $(\threesat, \nu)$ to $(\uepmc, \pi_1)$,
which yields the result by 
Theorem~\ref{thm:chopped-np-completeness}.
\end{proof}

\paragraph{\bf Acknowledgements.}
This work was supported by the Spanish project
TIN2013-46181-C2-2-R, 
by the Basque project GIU12/26,
and by the Basque grant UFI11/45.



\bibliography{../hubiebib}



\newpage
\appendix

\section{Proofs for Section~\ref{subsect:reduction}}
\label{sect:reduction-proofs}

We establish a lemma that will be of aid.

\begin{lemma}
\label{lemma:compose-poly-comp-functions}
Suppose that 
$(Q, \kappa)$ and $(Q', \kappa')$
are parameterized problems;
that $(g,s)$ is a poly-comp reduction
from  $(Q, \kappa)$ to $(Q', \kappa')$;
and,
that $g'$ is poly-compilable with respect to $\kappa'$
and $Q''$ is a language
with $x' \in Q' \Leftrightarrow g'(x') \in Q''$.
Then, there exists a function $g^+: \Sigma^* \to \Sigma^*$ 
that is poly-compilable with respect to $\kappa$
such that $x \in Q \Leftrightarrow g^+(x) \in Q''$.

If one assumes in addition 
that $\kappa''$ is a parameterization
and that $s': \Sigma^* \to \wpfin(\Sigma^*)$ is a function
such that $(g',s')$ is a poly-comp reduction
from $(Q', \kappa')$ to $(Q'', \kappa'')$,
then there exists $s^+: \Sigma^* \to \wpfin(\Sigma^*)$
such that $(g^+, s^+)$ is a poly-comp reduction
from $(Q, \kappa)$ to $(Q'', \kappa'')$.
\end{lemma}

\begin{proof}
Suppose that $g$ is poly-compilable (with respect to $\kappa$)
via $f: \Sigma^* \to \Sigma^*$ and
$c: \Sigma^* \to \Sigma^*$;
suppose that $g'$ is poly-compilable (with respect to $\kappa'$)
via $f': \Sigma^* \to \Sigma^*$ and
$c': \Sigma^* \to \Sigma^*$.

Define $c^+: \Sigma^* \to \Sigma^*$ by
$c^+(k) = (c(k), \{ (k', c'(k') ~|~ k' \in s(k) \})$;
it is straightforward to verify that this function
is computable and has poly-length.
Define $f^+: \Sigma^* \times \Sigma^* \to \Sigma^*$
to be the function computed the polynomial-time algorithm
that, given a pair $(a,x)$ of strings:
\begin{itemize}

\item sets $k = \kappa(x)$;
\item views $a$ as a string of the form
$(d, t)$ where $t$ is a set of pairs;
\item computes $x' = f(d,x)$;
\item computes $k' = \kappa'(x')$ (which is in $s(k)$);
and then, if there exists $d'$ such that $(k',d') \in t$,
computes and outputs $x'' = f'(c'(k'),x')$
(if not, the algorithm outputs the empty string).

\end{itemize}
Define the function $g^+$ by $g^+(x) = f^+(c^+(\kappa(x),x)$.

We have that $g^+$ is poly-compilable with respect to $\kappa$.
Let $x \in \Sigma^*$ be any string, and set $k = \kappa(x)$.
Given the pair $(c^+(k),x)$, the algorithm for $f^+$
will compute $x' = f(c(k),x)$,
will compute $k' = \kappa'(x')$,
and 
will output $x'' = f'(c'(k'),x')$.
We have that
$x \in Q$ iff $x' = g(x) \in Q'$ iff $x'' = g'(x') \in Q''$.

Suppose that $\kappa''$ and $s'$ are as described
in the lemma statement.
Define $s^+$ as $\cup_{k' \in s(k)} s'(k')$.
It is straightforward to verify that $s^+$
has the desired properties.
\end{proof}

\begin{proof}
(Theorem~\ref{thm:polycomp-compatible})
Suppose that $(Q, \kappa)$ poly-comp reduces
to $(Q', \kappa')$,
and that the problem $(Q', \kappa')$
is in $\polycomp{\fancyc}$.
By definition of $\polycomp{\fancyc}$,
there exists a function $g'$ that is poly-compilable
with respect to $\kappa'$ and a language $Q'' \in \fancyc$
such that (for each $x \in \Sigma^*$)
$x' \in Q' \Leftrightarrow g'(x') \in Q''$.
The function $g^+$
provided by Lemma~\ref{lemma:compose-poly-comp-functions}
witnesses that $(Q, \kappa) \in \polycomp{\fancyc}$.
\end{proof}

\begin{proof}
(Theorem~\ref{thm:polycomp-reducibility-transitive})
Lemma~\ref{lemma:compose-poly-comp-functions}
implies directly that if there exist
a poly-comp reduction
from $(Q, \kappa)$ to $(Q', \kappa')$
and
a poly-comp reduction
from $(Q', \kappa')$ to $(Q'', \kappa'')$,
then there exists a poly-comp reduction
from $(Q, \kappa)$ to $(Q'', \kappa'')$.
\end{proof}


\section{Proofs for Section~\ref{sect:cdls}}
\label{sect:proofs-cdls}

\subsection{Proof of Theorem~\ref{thm:mixed-equals-chopped}}
\begin{proof}
Suppose that $B$ is in $\mixed{\fancyc}$
via $f$ and $B' \in \fancyc$.
We have $(x,y) \in B \Leftrightarrow (f(x,\un(|y|)),y) \in B'$.
To show that $(B, \pi_1) \in \chopped{\fancyc}$,
define
$g(x,y) = h(c(\pi_1(x,y)), (x,y))$ as follows.
Define $c(x) = \{ (i,f(x,\un(i))) ~|~ i = 0, \ldots, p(|x|) \}$;
define $h$ to be a function computed by a polynomial-time algorithm
that, on an input $(d,(x,y))$
where $d$ has the form 
$\{ (i,a_i) ~|~ i=0,\ldots,p(|x|) \}$ and $|(x,y)| \leq p(|x|)$,
outputs $(a_{|y|}, y)$; 
and otherwise outputs a fixed string $z_N \notin B'$.
We claim that $(B, \pi_1)$ is in $\chopped{\fancyc}$
via $g$.

Let $(x,y) \in \Sigma^* \times \Sigma^*$.
If $|(x,y)| > p(|x|)$, then
$(x,y) \notin B$ and $g(x,y) = x_N \notin B'$.
If $|(x,y)| \leq p(|x|)$, then
$g(x,y) = h(c(x),(x,y))
= (f(x,\un(|y|)),y) \in B' \Leftrightarrow (x,y) \in B$.
Using the fact that $f$ is poly-length, 
it is straightforward to verify that
there exists a polynomial $p'$ such that,
it holds that $|g(x,y)| \leq p'(|\pi_1(x,y)|) = p'(|x|)$.

Now suppose that
$(B, \pi_1)$ is in $\chopped{\fancyc}$.
Then, there exists a function
$g(x,y) = h(c(x),(x,y))$
that is poly-comp with respect to $\pi_1$ 
and a language $D \in \fancyc$
such that 
$(x,y) \in B \Leftrightarrow g(x,y) \in D
\Leftrightarrow h(c(x),(x,y)) \in D$.
Define $D' = \{ ((e,x),y) ~|~ h(e,(x,y)) \in D \}$.
Clearly, $D'$ many-one polynomial-time reduces to $D$
(since $h$ is polynomial-time computable),
so by hypothesis, $D' \in \fancyc$.
Define $c'$ to be a function such that
$c'(x,1^k) = (c(x),x)$.
We have
$(x,y) \in B 
\Leftrightarrow
h(c(x),(x,y) \in D
\Leftrightarrow
((c(x),x),y) = (c'(x,\un(|y|)),y) \in D'$.
Hence, we have $B \in \mixed{\fancyc}$
via $c'$ and $D'$.
\end{proof}

\subsection{A lemma}

We use $\mu: \Sigma^* \times \Sigma^* \to \Sigma^*$
to denote the parameterization
defined by $\mu(x,y) = (x,\un(|y|))$.

\begin{lemma}
\label{lemma:mixed-to-polycomp}
Suppose that $A$ and $B$ are languages of pairs
such that there exists a mixed reduction from $A$ to $B$.
Then, there exists a poly-comp reduction
from $(A, \mu)$ to $(B, \pi_1)$.
\end{lemma}

\begin{proof}
Let $(f_1, f_2, g)$ be the mixed reduction
from $(A, \mu)$ to $(B, \pi_1)$.
We define a poly-comp reduction $(h,s)$
with $h(x,y) = f(c(\mu(x,y)), (x,y))$
defined as follows.
Define $c$ to be a computable, poly-length function such that
$c(x,1^k) = (f_1(x,1^k),f_2(x,1^k))$;
define $f$ so that
$f((a_1,a_2),(x,y)) = (a_1,g(a_2,y))$.
Define $s(x,1^k) = \{ f_1(x,1^k) \}$.

Let $(x,y) \in \Sigma^* \times \Sigma^*$.
Observe that 
$h(x,y) 
= f( (f_1(x,\un(|y|)),   f_2(x,\un(|y|))), (x,y))
=  (  f_1(x,\un(|y|)), g(f_2(x,\un(|y|))), y) )$.
We thus have
$(x,y) \in A \Leftrightarrow
h(x,y) \in B$.
Moreover,
we have $\pi_1(h(x,y)) = f_1(x,\un(|y|)) \in s(\mu(x,y))$.
\end{proof}

\subsection{Proof of Theorem~\ref{thm:hardness-from-epsilon}}

\begin{proof}
We give a poly-comp reduction
from $(Q^+,\len)$ to $(A, \mu)$.
By Lemma~\ref{lemma:mixed-to-polycomp}
(and Theorem~\ref{thm:polycomp-reducibility-transitive}),
this implies that there exists a poly-comp reduction
from $(Q^+,\len)$ to $(B, \pi_1)$,
which suffices by Proposition~\ref{prop:complete-for-choppedc}.
Define $g$ by $g(x) = (\epsilon,x)$
and
$s$ by $s(1^m) = \{ \mu(\epsilon,1^m) \}$.
It is straightforward to verify that
$(g,s)$ is a poly-comp reduction from $(Q^+,\len)$
to $(A, \mu)$.
\end{proof}

\subsection{Proof of Theorem~\ref{thm:hardness-from-poly-bd-slices}}

\begin{proof}
We give a poly-comp reduction from $(A, \pi_1)$ to $(A, \mu)$.
This suffices by Lemma~\ref{lemma:mixed-to-polycomp}
(and Theorem~\ref{thm:polycomp-reducibility-transitive}).
Fix $z_N$ to be a string not in $A$.
Let $p$ be a polynomial witnessing that $(A, \pi_1)$
has poly-bounded slices.
Define $g: \Sigma^* \times \Sigma^* \to \Sigma^*$
to be the function where
$g(x,y) = (x,y)$ if $|y| \leq p(|x|)$,
and $g(x,y) = z_N$ if $|y| > p(|x|)$.
Define $s: \Sigma^* \to \Sigma^*$
by $s(x) = \{ (x,\un(0)), (x,\un(1)), \ldots, (x,\un(p(|x|))) \}
\cup \{ \mu(z_N) \}$.
It is straightforward to verify that $(g,s)$
is a poly-comp reduction from $(A, \pi_1)$ to $(A, \mu)$.
\end{proof}

\end{document}